\documentclass{elsarticle-arxiv}
\usepackage[english]{babel}

\usepackage{relsize,xspace}
 \usepackage{xcolor}
 \usepackage{mathtools}
 \usepackage{todonotes}
 \usepackage{comment}
\usepackage{microtype}
\usepackage{amsmath}
\usepackage{amssymb}
\usepackage{amsfonts}
\usepackage{stmaryrd}
\usepackage{bm}
\usepackage{tikz}
\usepackage{refcount}
\usepackage{wrapfig}
\usepackage{thm-restate}
\usepackage{mathrsfs}
\usepackage[absolute]{textpos}
\usepackage[all]{nowidow}

\definecolor{blue}{rgb}{0.1,0.2,0.5}
\definecolor{brown}{rgb}{0.6,0.6,0.2}
\usepackage[ocgcolorlinks, linkcolor={blue}, citecolor={brown}]{hyperref}

\usepackage[amsmath,thmmarks,hyperref]{ntheorem}
\usepackage{cleveref}

\crefformat{page}{#2page~#1#3}%
\Crefformat{page}{#2Page~#1#3}%
\crefformat{equation}{#2(#1)#3}%
\Crefformat{equation}{#2(#1)#3}%
\crefformat{figure}{#2Figure~#1#3}%
\Crefformat{figure}{#2Figure~#1#3}%
\crefformat{section}{#2Section~#1#3}
\Crefformat{section}{#2Section~#1#3}
\crefformat{chapter}{#2Chapter~#1#3}
\Crefformat{chapter}{#2Chapter~#1#3}
\crefformat{chapter*}{#2Chapter~#1#3}
\Crefformat{chapter*}{#2Chapter~#1#3}
\crefformat{part}{#2Part~#1#3}
\Crefformat{part}{#2Part~#1#3}
\crefformat{enumi}{#2(#1)#3}
\Crefformat{enumi}{#2(#1)#3}

\usepackage{enumerate}

\usepackage{latexsym}

\theoremnumbering{arabic}
\theoremstyle{plain}
\theoremsymbol{}
\theorembodyfont{\itshape}
\theoremheaderfont{\normalfont\bfseries}
\theoremseparator{.}
\theorempreskip{4pt}
\theorempostskip{3pt}

\newtheorem{theorem}{Theorem}
\crefformat{theorem}{#2Theorem~#1#3}
\Crefformat{theorem}{#2Theorem~#1#3}

\newcommand{\newtheoremwithcrefformat}[2]{%
  \newtheorem{#1}[theorem]{#2}%
  \crefformat{#1}{##2\MakeUppercase#1~##1##3}%
  \Crefformat{#1}{##2\MakeUppercase#1~##1##3}%
}
\newcommand{\newseptheoremwithcrefformat}[2]{%
  \newtheorem{#1}{#2}%
  \crefformat{#1}{##2\MakeUppercase#1~##1##3}%
  \Crefformat{#1}{##2\MakeUppercase#1~##1##3}%
}

\newtheoremwithcrefformat{lemma}{Lemma}
\newtheoremwithcrefformat{proposition}{Proposition}
\newtheoremwithcrefformat{observation}{Observation}
\newtheoremwithcrefformat{conjecture}{Conjecture}
\newtheoremwithcrefformat{corollary}{Corollary}
\newseptheoremwithcrefformat{claim}{Claim}
\theorembodyfont{\upshape}
\newtheoremwithcrefformat{example}{Example}
\newtheoremwithcrefformat{remark}{Remark}
\newseptheoremwithcrefformat{definition}{Definition}

\theoremstyle{nonumberplain}
\theoremheaderfont{\scshape}
\theorembodyfont{\normalfont}
\theoremsymbol{\ensuremath{\square}}
\newtheorem{proof}{Proof}

\theoremsymbol{\ensuremath{\lrcorner}}
\newtheorem{clproof}{Proof}

\def\cqedsymbol{\ifmmode$\lrcorner$\else{\unskip\nobreak\hfil
\penalty50\hskip1em\null\nobreak\hfil$\lrcorner$
\parfillskip=0pt\finalhyphendemerits=0\endgraf}\fi}

\newcommand{\wcol}{\mathrm{wcol}}

\newcommand{\WReach}{\mathrm{WReach}}

\newcommand{\Oof}{\mathcal{O}}
\newcommand{\CCC}{\mathscr{C}}

\newcommand{\cutrk}{\mathrm{cutrk}}

\newcommand{\rw}{\mathrm{rw}}

\newcommand{\td}{\mathrm{td}}

\newcommand{\N}{\mathbb{N}}

\renewcommand{\phi}{\varphi}
\renewcommand{\epsilon}{\varepsilon}

\newcommand{\dist}{\mathrm{dist}}

\newcommand{\abs}[1]{\ensuremath{\left\lvert#1\right\rvert}}

\journal{European Journal of Combinatorics}

\title{On Low Rank-Width Colorings}

\author{O-joung Kwon\tnoteref{t0}}
\address{Department of Mathematics, Incheon National University, Incheon, South Korea
\\[2pt]
\texttt{ojoungkwon@gmail.com}}
\ead{ojoungkwon@gmail.com}
\author{Micha\l{} Pilipczuk}
\address{Institute of Informatics, University of Warsaw, Poland\\[2pt]
\texttt{michal.pilipczuk@mimuw.edu.pl}}
\ead{michal.pilipczuk@mimuw.edu.pl}
\author{Sebastian Siebertz\tnoteref{t2}}
\address{Institut f\"ur Informatik, 
Humboldt-Universit\"at zu Berlin, Germany\\[2pt]
\texttt{sebastian.siebertz@hu-berlin.de}}
\ead{sebastian.siebertz@hu-berlin.de}

\tnotetext[t0]{A preliminary version of this work appeared in the
  proceeding of WG 2017~\cite{kwon17}. This full version contains
  complete proofs of all the claimed results. This work was done while
  S. Siebertz was affiliated with University of Warsaw. M. Pilipczuk
  and S. Siebertz were supported by the National Science Centre of
  Poland via POLONEZ grant agreement UMO-2015/19/P/ST6/03998, which
  has received funding from the European Union's Horizon 2020 research
  and innovation programme (Marie Sk\l odowska-Curie grant agreement
  No.\ 665778). M. Pilipczuk was also supported by the Foundation for
  Polish Science (FNP) via the START stipend programme.}
\tnotetext[t2]{O. Kwon is supported by the National Research
  Foundation of Korea (NRF) grant funded by the Ministry of Education
  (No. NRF-2018R1D1A1B07050294).  }

\begin{document}

\begin{frontmatter}
\begin{abstract}
  We introduce the concept of \emph{low rank-width colorings},
  generalizing the notion of low tree-depth colorings introduced by
  Ne{\v{s}}et{\v{r}}il and Ossona~de Mendez in [Grad and classes with
  bounded expansion {I}. {D}ecompositions. Eur. J. Comb., 2008].
  We say that a class $\CCC$ of graphs admits \emph{low rank-width
    colorings} if there exist functions $N\colon \N\rightarrow\N$ and
  $Q\colon \N\rightarrow\N$ such that for all $p\in \N$, every graph
  $G\in \CCC$ can be vertex colored with at most $N(p)$ colors such
  that the union of any $i\leq p$ color classes induces a subgraph of
  rank-width at most $Q(i)$.

  Graph classes admitting low rank-width colorings strictly generalize
  graph classes admitting low tree-depth colorings and graph classes
  of bounded rank-width.  We prove that for every graph class $\CCC$
  of bounded expansion and every positive integer $r$, the class
  $\{G^r\colon G\in \CCC\}$ of $r$th powers of graphs from~$\CCC$
  admits low rank-width colorings.
  On the negative side, we show that the classes of interval graphs
  and permutation graphs do not admit low rank-width colorings.  As
  interesting side properties, we prove that every hereditary graph
  class admitting low rank-width colorings has the Erd\H{o}s-Hajnal
  property and is $\chi$-bounded.
\end{abstract}

\begin{textblock}{5}(11.13, 13.45)
\includegraphics[width=38px]{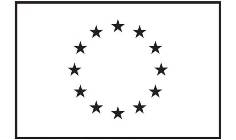}%
\end{textblock}
\end{frontmatter}

\section{Introduction and main results}

\noindent
A proper vertex coloring of a graph $G$ is an assignment of colors to
the vertices of~$G$ so that no two adjacent vertices are assigned the
same color, or in other words, so that every single color class
induces an independent set. In this simplest form, a vertex coloring
does not have to satisfy any constraints on the interaction between
two or more color classes. By adding such constraints we may be able
to derive a lot more structural information about a graph than from a
simple proper vertex coloring. For example, DeVos et
al.~\cite{devos2004excluding} introduced colorings, later called
\emph{$p$-tree-width colorings}, that, for a fixed integer parameter
$p$, have the property that the union of any $i\leq p$ color classes
induces a subgraph of tree-width at most $i-1$. In such a coloring,
every single color class induces a graph of tree-width~$0$, that is,
it is an independent set. Hence low tree-width colorings are in particular
proper vertex colorings.  The union of any two color classes induces a
graph of tree-width $1$, that is, a forest, the union of any three
color classes induces a graph of tree-width $2$, and so on. Using the
structure theorem of Robertson and Seymour~\cite{robertson2003graph}
for graphs excluding a fixed graph as a minor, De Vos et al.\ proved
that for every graph~$H$ and every integer $p\geq 1$, there is an
integer $N=N(H,p)$, such that every $H$-minor-free graph $G$ admits a
$p$-tree-width coloring with $N$ colors.

Despite the graph theoretic interest in such colorings, we also find
nice algorithmic applications. Consider e.g.\ the subgraph isomorphism
problem. Here, we are given two graphs~$G$ and~$H$ as input, and we
are asked to determine whether~$G$ contains a subgraph isomorphic to
$H$. In many natural settings the pattern graph $H$ we are looking for
is small and in such case a low-tree-width coloring as described above
is very useful. Assume $H$ has $p$ vertices and assume we can
efficiently compute a $p$-tree-width coloring of $G$ with $N$
colors. We will find a copy of~$H$ in $G$, if it exists, in at most
$p$ color classes. We can hence iterate through all combinations of at
most $p$ colors and, by the properties of the coloring, look for $H$
in a subgraph of $G$ of tree-width at most $p-1$. From an algorithmic
point of view, graphs of bounded tree-width are very well behaved
graphs.  Many NP-complete problems, in fact, all problems that can be
formulated in monadic second order logic, are solvable in linear time
on graphs of bounded
tree-width~\cite{courcelle1990graph,courcelle1990monadic}. In
particular, the subgraph isomorphism problem for every fixed pattern
graph~$H$ is solvable in linear time on any graph of bounded
tree-width, in our case, in time $f(p)\cdot n$ for some computable
function $f$. Hence, if we have already computed a $p$-tree-width
coloring of $G$ with $N$ colors, we need only an additional running
time of $\Oof(N^p\cdot f(p)\cdot n)$ to solve the subgraph isomorphism
problem.

%

\emph{Tree-depth} is another important and useful graph invariant. It
was introduced under this name in~\cite{nevsetvril2006tree}, but
equivalent notions were known before, including the notion of
\emph{rank}~\cite{nevsetvril2003order}, \emph{vertex ranking number}
and minimum height of an \emph{elimination
  tree}~\cite{bodlaender1998rankings,deogun1994vertex,schaffer1989optimal},
etc.  In~\cite{nevsetvril2006tree}, Ne\v{s}et\v{r}il and Ossona de
Mendez introduced the notion of 
\emph{$p$-tree-depth colorings} as vertex
colorings of a graph such that the union of any $i\leq p$ color
classes induces a subgraph of tree-depth at most $i$.  Note that the
tree-depth of a graph is always larger (at least by $1$) than its
tree-width, hence a low tree-depth coloring is a stronger requirement
than a low tree-width coloring.  Also based on the structure theorem,
Ne\v{s}et\v{r}il and Ossona de Mendez~\cite{nevsetvril2006tree} proved
that proper minor-closed classes admit even low tree-depth colorings.

Not much later, Ne\v{s}et\v{r}il and Ossona de
Mendez~\cite{nevsetvril2008grad} proved that proper minor-closed
classes are unnecessarily restrictive for the existence of low
tree-depth colorings. They introduced the notion of \emph{bounded
  expansion classes of graphs}, a concept that generalizes the concept
of classes with excluded minors and with excluded topological
minors. While the original definition of bounded expansion is in terms
of density of shallow minors, it turns out that low tree-depth colorings
give an alternative characterisation: a class $\CCC$ of graphs has
bounded expansion if and only if for all $p\in \N$ there exists an
integer $N=N(\CCC, p)$ such that every graph $G\in \CCC$ admits a
$p$-tree-depth coloring with $N$ colors~\cite{nevsetvril2008grad}.
For the even more general notion of \emph{nowhere dense classes of
  graphs}~\cite{nevsetvril2011nowhere}, it turns out that a class
$\CCC$ of graphs closed under taking subgraphs is nowhere dense if and
only if for every $p\in \N$ and every $\epsilon>0$ there exists $n_0$ such
that every $n$-vertex graph $G\in \CCC$ with $n\geq n_0$ admits a
$p$-tree-depth coloring with~$n^\epsilon$ colors.

Furthermore, there is a simple algorithm to compute such a
coloring in time~$\Oof(n)$ in case~$\CCC$ has bounded expansion
and in time $\Oof(n^{1+\epsilon})$ for any $\epsilon>0$ in 
case $\CCC$
is nowhere dense.  As a result, the subgraph isomorphism problem for
every fixed pattern $H$ can be solved in linear time on any class of
bounded expansion and in almost linear time on any nowhere dense
class. More generally, it was shown
in~\cite{dvovrak2013testing,grohe2011methods} that every fixed first
order property can be tested in linear time on graphs of bounded
expansion, implicitly using the notion of low tree-depth colorings,
and in almost linear time on nowhere dense
classes~\cite{grohe2014deciding}.

Note that bounded expansion and nowhere dense classes of graphs are
uniformly sparse graphs. In fact, bounded expansion classes of graphs
can have at most a linear number of edges and nowhere dense classes
can have no more than $\Oof(n^{1+\epsilon})$ many edges, for any fixed $\epsilon>0$.  This
motivates our new definition of \emph{low rank-width colorings} which
extends the coloring technique to dense classes of graphs which are
closed under taking induced subgraphs.

Rank-width was introduced by Oum and
Seymour~\cite{oum2006approximating} and aims to extend tree-width by
allowing well behaved dense graphs to have small rank-width. Also for
graphs of bounded rank-width there are many efficient algorithms based
on dynamic programming.  Here, we have the important meta-theorem of
Courcelle, Makowsky, and Rotics~\cite{courcelle2000linear}, stating
that for every monadic second-order formula (with set quantifiers
ranging over sets of vertices) and every positive integer~$k$, there
is an $\Oof(n^3)$-time algorithm to determine whether an input graph
of rank-width at most $k$ satisfies the formula.  There are several
parameters which are equivalent to rank-width in the sense that one is
bounded if and only if the other is bounded. These include
\emph{clique-width}~\cite{courcelle1993handle},
\emph{NLC-width}~\cite{wanke1994k}, and
\emph{Boolean-width}~\cite{bui2011boolean}.

\pagebreak
\paragraph*{Low rank-width colorings} We now introduce our main object
of study.

\begin{definition}
  A class $\CCC$ of graphs \emph{admits low rank-width colorings} if
  there exist functions $N:\N\rightarrow\N$ and $Q:\N\rightarrow\N$
  such that for all $p\in \N$, every graph $G\in \CCC$ can be vertex
  colored with at most $N(p)$ colors such that the union of any
  $i\leq p$ color classes induces a subgraph of rank-width at most
  $Q(i)$.
\end{definition}

As proved by Oum~\cite{oum2008rank}, every graph $G$ with tree-width
$k$ has rank-width at most $k+1$, hence every graph class which admits
low tree-width colorings (that is, every class of bounded expansion)
also admits low rank-width colorings.  On the
other hand, graphs admitting a low rank-width coloring can be very
dense.  We also remark that graph classes admitting low rank-width
colorings can be closed under taking induced subgraphs without spoiling this property, as rank-width
does not increase by removing vertices.

Let us remark that due to the model-checking algorithm of Courcelle et
al.~\cite{courcelle2000linear}, the (induced) subgraph isomorphism
problem is solvable in cubic time for every fixed pattern $H$ whenever
the input graph is given together with a low rank-width coloring for
$p=|V(H)|$, using $N(p)$ colors.  Again, it suffices to iterate
through all $p$-tuples of color classes and look for the pattern $H$
in the subgraph induced by these color classes; this can be done
efficiently since this subgraph has rank-width at most $Q(p)$. The
caveat is that the graph has to be supplied with an appropriate
coloring.  In this work we do not investigate the algorithmic aspects
of low rank-width colorings, and rather concentrate on the
combinatorial question of which classes admit such colorings, and
which do not.

\paragraph*{Related concepts}
Let us discuss the relation of our newly introduced
notion of low rank-width colorings with other related concepts studied
in the literature. The concept of \emph{bounded local 
clique-width} was introduced by Grohe and Tur\'an 
in~\cite{grohe2004learnability}. A class $\CCC$ of graphs
has bounded local clique-width if there is a function $f\colon \N\times
\N$ such that for every $r\geq 1$, $G\in\CCC$ and $v\in V(G)$, 
the clique-width of the subgraph of $G$ 
induced by the $r$-neighborhood of
$v$, $G[N_r(v)]$, is bounded by~$f(r)$. A more restricted 
concept, called \emph{nice local clique-width decompositions},
was introduced by Courcelle, Gavoille, and
Kant\'e~\cite{CourcelleGK2011}. They defined a \emph{nice
  $(r, \ell, g)$-cover} of a graph $G$ as a family $\mathcal{F}$ of
subsets of $V(G)$ which covers $G$ such that (1)
the $r$-neighborhood of every vertex is contained in some set in~$\mathcal{F}$, (2)~each set of $\mathcal{F}$ intersects at 
most $\ell$
other sets in $\mathcal{F}$, and (3) any union of $q$ sets in~$\mathcal{F}$ induces a subgraph of clique-width at most $g(q)$.  A
class $\CCC$ of graphs is \emph{nicely locally clique-width
  decomposable} if for each $r$, every graph in $\CCC$ admits a
nice $(r, \ell, g)$-cover for some~$\ell$ and~$g$. 
It is implicit in~\cite{CourcelleGK2011} that every nicely
locally clique-width decomposable class admits low rank-width
colorings (this terminology is not defined 
in~\cite{CourcelleGK2011} and we will provide an explicit
proof below). Also every nicely locally clique-width
decomposable class has bounded local clique-width, while there
are classes of bounded local clique-width that do not admit
low rank-width decompositions. Also not every class that admits
low rank-width decompositions admits nice local clique-width
decompositions. In particular, both the concept of nice local
clique-width decompositions and the concept of bounded local
clique-width are not robust under small non-local changes, such
as adding a universal vertex, that is, a vertex connected to all
other vertices. 

On the other hand, our newly introduced notion of
low rank-width colorings is robust under this operation, as we can
simply assign a new unique color to the newly added vertex and adding
one vertex can increase the rank-width of any graph by at most 
one. In particular, our notion of low rank-width colorable 
graphs generalizes the notion of low tree-depth colorable graphs, 
and hence the notion of bounded expansion, while the notion 
of Coucelle, Gavoille, and Kant\'e  does not even
generalize the notion of $H$-minor-free graphs.

\paragraph*{Our contribution} We study for several 
graph classes whether they admit low rank-width covers
or not. In \cref{sec:other-positive}, we make explicit the 
fact that graph classes that admit nice local clique-width 
decompositions admit low rank-width colorings. 
Courcelle, Gavoille, 
and Kant\'e also proved in~\cite{CourcelleGK2011} that the class of unit
interval graphs and the class of bipartite permutation graphs
are nicely locally clique-width decomposable. We give a 
short direct proof for the existence of low rank-width colorings for these classes.

In \cref{sec:bounded-expansion} we present our main technical
contribution and prove that for every class $\CCC$ of 
bounded expansion and every
integer $r\geq 2$, the class $\{G^r\colon G\in \CCC\}$ of $r$th powers
of graphs from $\CCC$ admits low rank-width colorings.  It is easy to
see that there are classes of bounded expansion such that
$\{G^r\colon G\in \CCC\}$ has both unbounded rank-width and does not
admit low tree-depth colorings, hence our notion generalizes both of
these concepts. Our results inspired further research, and in a
follow-up of this work~\cite{Gaj18} we proved that in fact every
first-order interpretation of a bounded expansion class (the $r$th
power of a graph is a simple first-order interpretation) has low
rank-width colorings (and in fact low shrub-depth colorings). We
comment on these new results in the Conclusion,
\cref{sec:conc}.

On the negative side, in \cref{sec:negative-results}, we show
that the classes of interval graphs and of permutation graphs do not
admit low rank-width colorings.  To show it, we give a general
construction called \emph{twisted chain graphs}, and show that if a
graph class contains a twisted chain graph of order $n$ for every
large $n$, then it does not admit low rank-width colorings.

Lastly, in \cref{sec:EHchi}, we show that every
hereditary graph class admitting low rank-width colorings has the
Erd\H{o}s-Hajnal property~\cite{ErdosH1989} and is
$\chi$-bounded~\cite{Gyarfas1987}.  Erd\H{o}s and
Hajnal~\cite{ErdosH1989} conjectured that for every graph $H$, there
is $\epsilon=\epsilon(H)>0$ such that every $n$-vertex graph having no
$H$ induced subgraph has either an independent set or a clique of size
$n^{\epsilon}$.  A graph class with this property is said to have
the Erd\H{o}s-Hajnal property. We refer to~\cite{Liebenau2019,
Chudnovsky2018, FoxPS2017, Chudnovsky2014, BonamyBT2016}
for background and recent progress on the Erd\H{o}s-Hajnal 
conjecture. 
Gy\'arf\'as~\cite{Gyarfas1987} introduced the notion of 
$\chi$-boundedness as a relaxation of the concept of perfect 
graphs. For background we refer to a survey of $\chi$-bounded 
classes by Scott and Seymour~\cite{surveychi}.



\section{Preliminaries}\label{sec:prelim}

All graphs in this paper are finite, undirected and simple, that is,
they do not have loops or parallel edges. Our notation is standard, we
refer to~\cite{diestel2012graph} for more background on graph theory.
We write $V(G)$ for the vertex set of a graph~$G$ and $E(G)$ for its
edge set.  A \emph{vertex coloring} of a graph $G$ with colors from~$S$ is a mapping $c\colon V(G)\rightarrow S$.  For each $v\in V(G)$,
we call $c(v)$ the color of $v$.  The {\em{distance}} between vertices
$u$ and $v$ in $G$, denoted $\dist_G(u,v)$, is the length of a
shortest path between~$u$ and $v$ in $G$, or $\infty$ if no
such path exists.  The \emph{$r$th power of a
  graph $G$} is the graph~$G^r$ with vertex set $V(G)$, where there is
an edge between two vertices $u$ and $v$ in $G^r$ if and only if their
distance in $G$ is at most~$r$.

Rank-width was introduced by Oum and
Seymour~\cite{oum2006approximating}. We refer to the
surveys~\cite{hlinveny2008width,oum2016rank} for more background.  For
a graph $G$, we denote the adjacency matrix of $G$ by $A_G$, where for
$x,y\in V(G)$, $A_G[x,y]=1$ if and only if $x$ is adjacent to $y$.
Let $G$ be a graph. We define the \emph{cut-rank} function
$\cutrk_{G} \colon 2^{V} \rightarrow \N$ such that $\cutrk_G(X)$ is
the rank of the matrix $A_G[X, V(G)\setminus X]$ over the field 
$GF(2)$
(if $X=\emptyset$ or $X=V(G)$, then we let $\cutrk_G(X)=0$).

A \emph{rank-decomposition} of $G$ is a pair $(T,L)$, where $T$ is a
subcubic tree (i.e.\ a tree where every node has degree $1$ or $3$)
with at least $2$ nodes and $L$ is a bijection from $V(G)$ to the set
of leaves of $T$.
For every edge $e\in E(T)$, $T-e$ has exactly two components 
$T_1, T_2$, and hence $e$ induces a vertex bipartition 
$(A^e_1, A^e_2)$ of $G$, where $A^e_i=L^{-1}(T_i)$. 
The \emph{width} of $e$ is defined as $\cutrk_G(A^e_1)$ 
(which is equal to $\cutrk_G(A^e_2)$). The \emph{width} 
of $(T,L)$ is the maximum width
over all edges in $T$, and the \emph{rank-width} of $G$, denoted by
$\rw(G)$, is the minimum width over all rank-decompositions of $G$.
If $\abs{V(G)}\le 1$, then~$G$ has no rank-decompositions, and the
rank-width of $G$ is defined to be $0$.

A \emph{tree-decomposition} of a graph $G$ is a pair $(T,\mathcal{B})$
consisting of a tree $T$ and a family
$\mathcal{B}=\{B_t\}_{t\in V(T)}$ of sets $B_t\subseteq V(G)$,
satisfying the following three conditions:
\begin{enumerate}[(T1)]
\item $V(G)=\bigcup_{t\in V(T)}B_t$;
\item for every $uv\in E(G)$, there exists a node $t$ of $T$ such that
  $\{u,v\}\subseteq B_t$;
\item for $t_1,t_2,t_3\in V(T)$,
  $B_{t_1}\cap B_{t_3}\subseteq B_{t_2}$ whenever $t_2$ is on the path
  from $t_1$ to $t_3$ in $T$.
\end{enumerate}
The \emph{width} of a tree-decomposition $(T,\mathcal{B})$ is
$\max\{ \abs{B_{t}}-1:t\in V(T)\}$.  The \emph{tree-width} of $G$ is
the minimum width over all tree-decompositions of $G$.

Let $G$ be a graph and let $G_1,\ldots, G_s$ be its connected
components.  Then the \emph{tree-depth} of $G$ is recursively defined
as
\[\td(G) = \begin{cases}
  1 & \text{if $|V(G)|=1$}\\
  1+\min_{v\in V(G)}\ \td(G-v) & \text{if $|V(G)|>1$ and $s=1$}\\
  \max_{1\leq i\leq s}\ \td(G_i) & \text{otherwise}.
\end{cases}\]

A graph is an \emph{interval graph} if it is the intersection graph of
a family $\mathcal{I}$ of intervals on the real line.  An interval
graph is a \emph{unit interval graph} if all intervals in
$\mathcal{I}$ have the same length.  A graph is a \emph{permutation
  graph} if it is the intersection graph of line segments whose
endpoints lie on two parallel lines.

\section{Nice local clique-width decompositions}\label{sec:other-positive}

\noindent In this section we first study the connections 
between the newly introduced notion of low rank-width colorings and 
the notion of nice local clique-width decompositions 
that was introduced by Courcelle, Gavoille, and 
Kant\'e~\cite{CourcelleGK2011}. We show that graph classes
that admit nice local clique-width decompositions admit
low rank-width colorings. We then prove that 
unit interval graphs and bipartite
permutation graphs admit low rank-width colorings. As discussed in the
introduction, this can be obtained by the result that these
classes are nicely locally clique-width decomposable, as shown by
Courcelle, Gavoille, and Kant\'e~\cite{CourcelleGK2011}.  On the other
hand, it is a simple application of Lozin's characterisations of those
two classes; so we would like to present direct proofs.  

\medskip

Recall that for integers $r,\ell\geq 1$ and a function $g\colon \N\rightarrow \N$, a family $\mathcal{F}$ of vertex subsets of $G$ is 
called a \emph{nice $(r, \ell, g)$-cover} if 
\begin{itemize}
\item $\bigcup_{S\in \mathcal{F}} S=V(G)$,
\item for every vertex $v$ of $G$, the $r$-neighborhood of $v$ is
  contained in some set of $\mathcal{F}$,
\item each set of $\mathcal{F}$ intersects at most $\ell$ other sets
  in $\mathcal{F}$,
\item any union of $q$ sets in $\mathcal{F}$ induces a subgraph of
  clique-width at most $g(q)$.
\end{itemize} 
A class $\mathcal{C}$ of graphs is \emph{nicely locally clique-width
  decomposable} if for each $r$, every graph in $\mathcal{C}$ admits a
nice $(r, \ell, g)$-cover for some $\ell$ and $g$.

\begin{theorem}[Implicit in~\cite{CourcelleGK2011}]\label{thm:CGK2011}
  Every nicely locally clique-width decomposable class admits low
  rank-width colorings.
\end{theorem} 
\begin{proof}
  Let $\CCC$ be a nicely locally clique-width decomposable
  class. Consider any $G\in\CCC$ and let $p\geq 1$ be an integer. We want
  to show that $G$ admits a low rank-width coloring with 
  $N=N(p)$ colors so that the union of any $i\leq p$ color classes
  induces a subgraph of rank-width at most $Q=Q(i)$ for 
  constants~$N$ and~$Q$ depending only on $p$ and $i$. 
  By assumption,
  $G$ admits a nice $(1, \ell, g)$-cover~$\mathcal{F}$ for 
  a constant~$\ell$ and function~$g$. Without loss of generality
  assume that $\ell\geq 3$. 

  As the neighborhood of each $v\in V(G)$ is contained in some
  set of $\mathcal{F}$, we can assign a mapping
  $f\colon V(G)\rightarrow \mathcal{F}$ such that the 
  neighborhood of $v$ is fully contained in $f(v)$.  

  Now, we define an auxiliary graph $H$ on vertex set
  $\mathcal{F}$ such that for $P,Q\in \mathcal{F}$, 
  $P$ and $Q$ are adjacent in $H$
  if and only if $P\cap Q\neq \emptyset$. As $\mathcal{F}$ is a nice
  $(1, \ell, g)$-cover, $H$ has maximum degree at most $\ell$. We
  conclude that the $p$-th power $H^p$ of $H$ has maximum 
  degree bounded by $\sum_{i=0}^p\ell^i<
  1+(p+1)\cdot \ell^p\eqqcolon N(p)$.  
  Thus, $H^p$ admits a proper coloring $\lambda$ with 
  $N$ colors. We define a coloring $\eta$ of $G$ by assigning 
  $v\in V(G)$ the color $\lambda(f(v))$. 

  We let $Q(i)\coloneqq g(i)$ and show that the union of 
  any $i\leq p$ color classes
  induces a subgraph of rank-width at most $Q(i)$;
  as $G$ and $p$ were chosen arbitrarily, this means that $\CCC$ admits low rank-width colorings.
  For this, it suffices to prove the following claim:
  if a vertex subset $X\subseteq V(G)$ receives at most $i$ different colors in the coloring $\eta$, then $G[X]$ has rank-width at most $Q(i)$.
  Since the rank-width of a disjoint union of graphs of rank-width at most $Q(i)$ also has rank-width at most $Q(i)$, we may focus on the case when $G[X]$ is connected.
  
  Let $\mathcal{X}=\{f(v)\colon v\in X\}\subseteq \mathcal{F}$.
  By the definition of $\eta$ we have that $\mathcal{X}$ receives at most $i$ different colors under $\lambda$.
  Note that whenever $uv$ is an edge in $G$, $f(u)$ and $f(v)$ are either equal or adjacent in $H$.
  Therefore, we conclude that since $G[X]$ is connected, $H[\mathcal{X}]$ is connected as well.
  
  Suppose for a moment that $|\mathcal{X}|>i$. Then we would be able to find a subset $\mathcal{X}'\subseteq \mathcal{X}$ with $|\mathcal{X}'|=i+1$ such that $H[\mathcal{X}']$ is connected a well.
  Observe that then the vertices of $\mathcal{X}'$ would be pairwise at distance at most $p$ in $H$, implying, by the construction of $\lambda$, that they would receive pairwise different colors in $\lambda$.
  This is a contradiction with the fact that $\mathcal{X}$ receives at most $i$ colors under $\lambda$, hence we conclude that $|\mathcal{X}|\leq i$.
  
  Since $\mathcal{F}$ is a $(1,\ell,g)$-cover, we conclude that the graph $G[\bigcup \mathcal{X}]$ has clique-width bounded by $g(i)$, hence also its rank-width is bounded by $g(i)$.
  As $G[X]$ is an induced subgraph of $G[\bigcup \mathcal{X}]$, we are done.
\end{proof}

Courcelle, Gavoille, and Kant\'e~\cite{CourcelleGK2011} showed that
the class of unit interval graphs is nicely locally clique-width
decomposable.  So, \cref{thm:CGK2011} implies that the class of
unit interval graphs admits low rank-width colorings. We 
nevertheless give an explicit proof of this observation. 

\begin{theorem}\label{thm:unit-interval}
  The class of unit interval graphs and the class of bipartite
  permutation graphs admit low rank-width colorings.
\end{theorem}

We recall the characterizations of these classes obtained by
Lozin~\cite{Lozin2011}.  Let \mbox{$n,m\geq 1$}. We denote by $H_{n,m}$ the
graph with $n\cdot m$ vertices which can be partitioned into $n$
independents sets
$V_1=\{v_{1,1},\ldots, v_{1,m}\}, \ldots, V_n=\{v_{n,1},\ldots,
v_{n,m}\}$
so that for each $i\in \{1,\ldots, n-1\}$ and for each
$j, j'\in \{1,\ldots, m\}$, vertex $v_{i,j}$ is adjacent to
$v_{i+1,j'}$ if and only if $j'\in \{1, \ldots, j\}$, and there are no
edges between~$V_i$ and~$V_j$ if $\abs{i-j}\ge 2$.  The graph
$\widetilde{H}_{n,m}$ is the graph obtained from $H_{n,m}$ by
replacing each independent set $V_i$ by a clique.

\begin{lemma}\label{lem:rwbound}
The following statements hold:
\begin{enumerate}
\item (Lozin~\cite{Lozin2011}) The rank-width of $H_{n,m}$ and of
  $\widetilde{H}_{n,m}$ is at most $3n$.
\item (Lozin~\cite{Lozin2011}) Every bipartite permutation graph on
  $n$ vertices is isomorphic to an induced subgraph of $H_{n,n}$.
\item (Lozin~\cite{LozinR2007}) Every unit interval graph on $n$
  vertices is isomorphic to an induced subgraph
  of~$\widetilde{H}_{n,n}$.
\end{enumerate}
\end{lemma}

Hence, in order to prove~\cref{thm:unit-interval}, it suffices to
prove that the graphs $H_{n,m}$ and $\widetilde{H}_{n,m}$ admit low
rank-width colorings.

\begin{proof}[of \cref{thm:unit-interval}]
  For every positive integer $p$, let $N(p):=p+1$ and $Q(i):=3i$ for
  each $i\in \{1, \ldots, p\}$.  We prove that for all $n,m\geq 1$,
  the graphs~$H_{n,m}$ and $\widetilde{H}_{n,m}$ can be vertex colored
  using $N(p)$ colors so that each of the connected components of the
  subgraph induced by any $i\le p$ color classes has rank-width at
  most $Q(i)$. As rank-width and rank-width colorings are monotone
  under taking induced subgraphs, the statement of the theorem follows
  from~\cref{lem:rwbound}.

  Assume that the vertices of $H_{n,m}$ (and $\widetilde{H}_{n,m}$,
  respectively) are, as in the definition, named
  $v_{1,1},\ldots, v_{1,m}, \ldots, v_{n,1},\ldots, v_{n,m}$.  We
  color the vertices in the $i$th row, $v_{i,1},\ldots, v_{i,m}$, with
  color $j+1$ where $j \in \{ 0, 1, \ldots, p\}$ and
  $i\equiv j {\pmod {p+1}}$. Then any connected component $H$ of a
  subgraph induced by $i\leq p$ colors is isomorphic to $H_{i',m}$
  ($\widetilde{H}_{i',m}$, respectively) for some $i'\leq i$.  Hence,
  according to~\cref{lem:rwbound}, $H$ has rank-width at most
  $3i=Q(i)$, as claimed.
\end{proof}

\section{Powers of sparse graphs}\label{sec:bounded-expansion}

\noindent In this section we show that the class of $r$th powers of graphs from
a bounded expansion class admit low rank-width colorings.
The original definition of bounded expansion classes by
Ne\v{s}et\v{r}il and Ossona de Mendez~\cite{nevsetvril2008grad} is in
terms of bounds on the density of bounded depth minors. We will work
with the characterisation by the existence of low tree-depth colorings
as well as by a characterisation in terms of bounds on generalized
coloring numbers.

\begin{theorem}[Ne\v{s}et\v{r}il and Ossona de
  Mendez~\cite{nevsetvril2008grad}]\label{thm:be-lcw}
  A class $\CCC$ of graphs has bounded expansion if and only if for
  all integers $p\geq 0$ there exists an integer $N=N(\CCC, p)$ such that every
  graph $G\in \CCC$ admits a $p$-tree-depth coloring with~$N$ colors.
\end{theorem}
Our main result in this section is the following.
\begin{theorem}\label{thm:be-lrw}
  Let $\CCC$ be a class of bounded expansion and $r\geq 2$ be an
  integer.  Then the class $\{G^r\colon G\in \CCC\}$ of $r$th powers
  of graphs from $\CCC$ admits low rank-width colorings.
\end{theorem}

For a graph $G$, we denote by $\Pi(G)$ the set of all linear orders of
$V(G)$.  For $u,v\in V(G)$ and an integer $r\geq 0$, we say
that~$u$ is \emph{weakly $r$-reachable} from~$v$ with respect to a
linear order $L$, if there is a path $P$ of length at most~$r$ 
between~$u$ and $v$ such that $u$ is the smallest among the 
vertices of $P$
with respect to~$L$.  We denote by $\WReach_r[G,L,v]$ the set of
vertices that are weakly $r$-reachable from~$v$ with respect to~$L$.
The \emph{weak $r$-coloring number $\wcol_r(G)$} of~$G$ is defined as
\begin{eqnarray*}
  \wcol_r(G)& := & \min_{L\in\Pi(G)}\:\max_{v\in V(G)}\:
                   \bigl|\WReach_r[G,L,v]\bigr|.
\end{eqnarray*}

The weak coloring number was introduced by Kierstead and
Yang~\cite{kierstead2003orders} in the context of coloring and marking
games on graphs.  As shown by Zhu~\cite{zhu2009coloring}, classes of
bounded expansion can be characterised by the weak coloring numbers.

\begin{theorem}[Zhu~\cite{zhu2009coloring}]\label{thm:charbddexp}
  A class $\CCC$ has bounded expansion if and only if for all
  integers $r\geq 0$ there is an integer $f(r)$ such that for all $G\in \CCC$
  we have $\wcol_r(G)\leq f(r)$.
\end{theorem}

In order to prove \cref{thm:be-lrw}, we will first compute a low
tree-depth coloring. We would like to apply the following theorem,
relating the tree-width (and hence in particular the tree-depth) of a
graph and the rank-width of its $r$th power.

\begin{theorem}[Gurski and Wanke~\cite{GurskiW2009}]\label{thm:tdtorw}
  Let $p\geq 0,r\ge 2$ be integers.  If a graph~$H$ has tree-width at most
  $p$, then $H^r$ has rank-width at most $2(r+1)^{p+1}-2$.
\end{theorem}

We remark that Gurski and Wanke~\cite{GurskiW2009} proved this bound
for clique-width instead of rank-width, but clique-width is never
smaller than the rank-width~\cite{oum2006approximating}.

The natural idea would be just to combine the bound of
\cref{thm:tdtorw} with low tree-depth coloring given by
\cref{thm:be-lcw}.  Note however, that when we consider any subgraph
$H$ induced by $i\leq p$ color classes, the graph $H^r$ may be
completely different from the graph $G^r[V(H)]$, due to paths that are
present in $G$ but disappear in $H$.  Hence we cannot directly apply
\cref{thm:tdtorw}. Instead, we will prove the existence of a refined
coloring of $G$ such that for any subgraph $H$ induced by $i\leq p$
color classes, in the refined coloring there is a subgraph $H'$ such
that $G^r[V(H)]\subseteq H'^r$ and such that $H'$ gets only $g(i)$
colors in the original coloring, for some fixed function $g$. We can
then apply \cref{thm:tdtorw} to $H'$ and use the fact that rank-width
is monotone under taking induced subgraphs.

In the following, we will say that a vertex subset $X$ {\em{receives}}
a color $i$ under a coloring $c$ if $i\in c^{-1}(X)$.  We first need
the following definitions.

\begin{definition}
  Let $G$ be a graph, $X\subseteq V(G)$ and $r\geq 2$.  A superset
  $X'\supseteq X$ is called \emph{an $r$-shortest path hitter} for $X$
  if for all $u,v\in X$ with $1<\dist_G(u,v)\leq r$, $X'$ contains an
  internal vertex of some shortest path between $u$ and $v$.
\end{definition}

\begin{definition}
  Let $G$ be a graph, let $c$ be a coloring of $G$, and let 
  $p\geq 1,r\geq 2$ and $d\geq 1$ be integers. A coloring $c'$ 
  is a \emph{$(d,r)$-good refinement} of
  $c$ if for every vertex set $X\subseteq V(G)$, there exists an
  $r$-shortest path hitter $X'$ of $X$ such that if~$X$ receives at
  most~$p$ colors under $c'$, then $X'$ receives at most $d\cdot p$
  colors under $c$.
\end{definition}


\noindent We use the weak coloring numbers to prove the existence of a
good refinement.

\begin{lemma}\label{lem:goodrefinement}
  Let $G$ be a graph and $k\geq 1,r\ge 2$ be integers.  Then every coloring
  $c$ of~$G$ using $k$ colors has a $(2\wcol_r(G), r)$-good refinement
  coloring using $k^{2\wcol_r(G)}$ colors.
\end{lemma}
\begin{proof}
  Let $\Gamma$ be the set of colors used by $c$, and let
  $d:=2\wcol_r(G)$.  The $(d,r)$-good refinement~$c'$ that we are
  going to construct will use subsets of $\Gamma$ of size at most $d$
  as the color set; the number of such subsets is at most
  $k^{2\wcol_r(G)}$.  Let $L$ be a linear order of $V(G)$ with
  $\max_{v\in V(G)}\: \bigl|\WReach_r[G,L,v]\bigr|=\wcol_r(G)$.  We
  construct a new coloring $c'$ as follows:
  \begin{enumerate}[(1)]
  \item Start by setting $c'(v):=\emptyset$ for each $v\in V(G)$.
  \item For each pair of vertices $u$ and $v$ such that
    $u\in \WReach_r[G, L, v]$, we add the color $c(u)$ to $c'(v)$.
  \item For each pair $u,v$ of non-adjacent vertices such that
    $u<_L v$ and $u\in \WReach_r[G,L,v]$, we do the following.  Check
    whether there is a path $P$ of length at most $r$ connecting $u$
    and $v$ such that all the internal vertices of $P$ are larger than
    both $u$ and $v$ in $L$.  If there is no such path, we do nothing
    for the pair $u,v$.  Otherwise, fix one such path $P$, chosen to
    be the shortest possible, and let $z$ be the vertex traversed by
    $P$ that is the largest in $L$.  Then add the color $c(z)$
    to~$c'(v)$.
  \end{enumerate}
  Thus, every vertex $v$ receives in total at most $2\wcol_{r}(G)$
  colors of $\Gamma$ to its final color $c'(v)$: at most
  $\wcol_{r}(G)$ in step (2), and at most $\wcol_{r}(G)$ in step (3),
  because we add at most one color per each $u\in \WReach_r[G,L,v]$.
  It follows that each final color $c'(v)$ is a subset of $\Gamma$ of
  size at most $2\wcol_{r}(G)=d$.

  We claim that $c'$ is a $(d,r)$-good refinement of $c$.  Let
  $X\subseteq V(G)$ be a set that receives at most $p$ colors under
  $c'$, say colors $A_1, \ldots, A_p\subseteq \Gamma$.  Let $X'$ be
  the set of vertices of $G$ that are colored by colors in
  $A_1\cup \cdots \cup A_p$ under $c$.  Since $|A_i|\leq d$ for each
  $i\in \{1,\ldots,p\}$, we have that $X'$ receives at most $d\cdot p$
  colors under $c$.
	
  To show that $X'$ is an $r$-shortest path hitter of $X$, let us
  choose any two vertices $u$ and $v$ in~$X$ with $u<_L v$ and
  $1<\dist_G(u,v)\leq r$.  If there is a shortest path from $u$ to $v$
  whose all internal vertices are larger than $u$ and $v$ in $L$, by
  step (3), $X'$ contains a vertex that is contained in one such path.
  Otherwise, a shortest path from $u$ to $v$ contains a vertex $z$
  with $L(z)<L(v)$ other than $u$ and~$v$.  This implies that there
  exists $z'\in \WReach_r[G,L,v]\setminus \{u\}$ on the path such that
  $c(z')\in c'(v)$, and hence $z'\in X'$ by step (2).  Therefore, $X'$
  is an $r$-shortest path hitter of $X$, as required.
\end{proof}

\begin{definition}
  Let $G$ be a graph, let $X\subseteq V(G)$, and let $r\geq 1$ be an
  integer.  A superset $X'\supseteq X$ is called \emph{an $r$-shortest
    path closure of $X$} if for each $u,v\in X$ with
  $\dist_G(u,v)=\ell\leq r$, $G[X']$ contains a path of length $\ell$
  between $u$ and $v$.
\end{definition}

\begin{definition}
  Let $G$ be a graph, let $c$ be a coloring of $G$, and let $r\geq 2$
  and $d\geq 1$ be integers.  A coloring~$c'$ is a \emph{$(d,r)$-excellent
    refinement} of $c$ if for every vertex set $X\subseteq V(G)$ there
  exists an $r$-shortest path closure $X'$ of $X$ such that if $X$
  receives~at most $p$ colors in $c'$, then $X'$ receives at most
  $d\cdot p$ colors in $c$.
\end{definition}

\noindent We inductively define excellent refinements from good
refinements.

\begin{lemma}\label{lem:excellentrefinement}
  Let $G$ be a graph, let $k\geq 1,r\ge 2$ be integers, and let
  $d_r\coloneqq \prod_{2\le \ell\le r}2\wcol_\ell(G)$.  Then every
  coloring~$c$ of $G$ using at most $k$ colors has a
  $(d_r, r)$-excellent refinement coloring using at most $k^{d_r}$
  colors.
\end{lemma}
\begin{proof}
  The proof is by induction on $r$.  For $r=2$, an $r$-shortest path
  hitter of a set~$X$ is an $r$-shortest path closure, and vice versa.
  Hence, the statement immediately follows from
  \cref{lem:goodrefinement}.  Now assume $r\ge 3$.  By induction
  hypothesis, there is a $(d_{r-1}, r-1)$-excellent refinement~$c_1$
  of $c$ with at most $k^{d_{r-1}}$ colors.  By
  applying~\cref{lem:goodrefinement} to~$c_1$, we obtain a
  $(2\wcol_r(G),r)$-good refinement $c'$ of~$c_1$ with at most
  $(k^{d_{r-1}})^{2\wcol_r(G)}=k^{d_r}$ colors.  We claim that $c'$ is
  a $(d_r,r)$-excellent refinement of $c$. Any set $X$ which gets at
  most $p$ colors from $c'$ can be first extended to an $r$-shortest
  path hitter $X'$ for $X$ which receives at most $2\wcol_r(G)\cdot p$
  colors. Then $X'$ can be extended by induction hypothesis to an
  $(r-1)$-shortest path closure $X''$ of $X'$ that receives at most
  $d_{r-1}\cdot 2\wcol_r(G)\cdot p=d_r\cdot p$ colors.
   
  It remains to show that $X''$ is an $r$-shortest path closure of
  $X$.  Take any $u,v\in X$ with $\dist_G(u,v)=\ell\leq r$. If
  $\ell\leq 1$, then $u,v$ are already adjacent in $G[X]$.  Otherwise,
  since $X'$ is an $r$-shortest path hitter for $X$, there is a vertex
  $z\in X'$ that lies on some shortest path connecting $u$ and $v$ in
  $G$.  In particular, $\dist_G(u,z)=\ell_1$ and $\dist_G(z,v)=\ell_2$
  for $\ell_1,\ell_2$ satisfying $\ell_1,\ell_2<\ell$ and
  $\ell_1+\ell_2=\ell$.  Since $X''$ is an $(r-1)$-shortest path
  closure of $X'$, we infer that $\dist_{G[X'']}(u,z)=\ell_1$ and
  $\dist_{G[X'']}(z,v)=\ell_2$.  Hence $\dist_{G[X'']}(u,v)=\ell$ by
  the triangle inequality.
\end{proof}

\begin{proof}[of~\cref{thm:be-lrw}]
  Let $G\in \CCC$ and let
  $d_r:=\prod_{2\le \ell\le r}2\wcol_\ell(G)$.  Since~$\CCC$ has
  bounded expansion, by \cref{thm:charbddexp}, for each $r$,
  $\wcol_r(G)$ is bounded by a constant only depending on $\CCC$.  We
  start by taking $c$ to be a $(d_r\cdot p)$-tree-depth coloring with
  $N(d_r\cdot p)$ colors, where $N$ is the function from
  \cref{thm:be-lcw}.  Then its $(d_r,r)$-excellent refinement $c'$ has
  the property that $c'$ uses at most $N(d_r\cdot p)^{d_r}$ colors,
  and every subset $X$ which receives at most $p$ colors in $c'$ has
  an $r$-shortest path closure $X'$ that receives at most $d_r\cdot p$
  colors in $c$.  Thus, the graph induced by $X$ in the $r$th power
  $G^r$ is the same as the graph induced by $X$ in the $r$th power
  $G[X']^r$.  Since $G[X']$ has tree-depth at most $d_r\cdot p$,
  by~\cref{thm:tdtorw}, $G[X']^r$ has rank-width at most
  $2(r+1)^{d_r\cdot p+1}-2$.  Therefore, $G^r[X]$ has rank-width at
  most $2(r+1)^{d_r\cdot p+1}-2$ as well.
\end{proof}

We now give two example applications of \cref{thm:be-lrw}.  A
\emph{map graph} is a graph that can be obtained from a 
plane graph by
making a vertex for each face, and adding an edge between two
vertices, if the corresponding faces share a vertex. 
We use the following characterization of map graphs as {\em{half-squares of planar graphs}}, due to Chen et al.~\cite{ChenGP02}.

\begin{lemma}[\cite{ChenGP02}]\label{lem:half-squares}
  Every map graph is an induced subgraph of the second power of a
  planar graph.
\end{lemma}

As the class of planar graphs is a bounded expansion class, we conclude
that map graphs admit low rank-width colorings.  A similar reasoning
can be performed for line graphs of graphs from any bounded expansion
graph class.  Thus, both map graphs and line graphs of graphs from any
fixed bounded expansion graph class admit low rank-width colorings.

\begin{lemma}
  If $\CCC$ is a graph class of bounded expansion, then there is a
  graph class of bounded expansion $\CCC_1$ such that all line graphs
  of graphs from $\CCC$ are induced subgraphs of graphs from
  $\CCC_1^2$.
\end{lemma}
\begin{proof}
  Observe that for any graph $G$, if by $G_1$ we denote 
  $G$ with every
  edge subdivided once, then the line graph of $G$ is a subgraph of
  $G_1^2$ induced by the subdividing vertices.  It follows that we may
  take $\CCC_1$ to be the class of all $1$-subdivisions of graphs from
  $\CCC$. This class also has bounded expansion.
\end{proof}


\section{Negative results}\label{sec:negative-results}

\noindent In contrast to the result in \cref{sec:other-positive}, we prove that
interval graphs and permutation graphs do not admit low rank-width
colorings. For this, we introduce twisted chain graphs.  Briefly, a
twisted chain graph $G$ consists of three vertex sets $A, B, C$ where
each of $G[A\cup C]$ and $G[B\cup C]$ is a chain graph, but the
ordering of $C$ with respect to the chain graphs $G[A\cup C]$ and
$G[B\cup C]$ are distinct. See Figure~\ref{fig:twistedchain} for an
illustration.

\begin{figure}
 \tikzstyle{v}=[circle, draw, solid, fill=black, inner sep=0pt, minimum width=3pt]
 \tikzstyle{w}=[rectangle, draw, solid, fill=black, inner sep=0pt, minimum width=5pt, minimum height=5pt]
  \centering
   \begin{tikzpicture}[scale=0.6]

        \node [v]  (v1) at (-.5, 0-1){};
        \node [v]  (v2) at (-.5, 2-1){};
        \node [v]  (v3) at (-.5, 4+1){};
        \node [v]  (v4) at (-.5, 6+1){};

        \node [v]  (w1) at (6.5, 0-1){};
        \node [v]  (w2) at (6.5, 2-1){};
        \node [v]  (w3) at (6.5, 4+1){};
        \node [v]  (w4) at (6.5, 6+1){};

        \node [v]  (z1) at (2, 2){};
        \node [v]  (z2) at (4, 2){};
        \node [v]  (z3) at (2, 4){};
        \node [v]  (z4) at (4, 4){};

	    \draw (v1) node [left] {$v_1$}; 
	    \draw (v2) node [left] {$v_2$}; 
	    \draw (v3) node [left] {$v_3$}; 
	    \draw (v4) node [left] {$v_4$}; 

	    \draw (w1) node [right] {$w_1$}; 
	    \draw (w2) node [right] {$w_2$}; 
	    \draw (w3) node [right] {$w_3$}; 
	    \draw (w4) node [right] {$w_4$}; 

	    \draw (z1) node [below] {$z_{1,1}$}; 
	    \draw (z2) node [below] {$z_{1,2}$}; 
	    \draw (z3) node [above] {$z_{2,1}$}; 
	    \draw (z4) node [above] {$z_{2,2}$}; 
	
		\draw(v1)--(z1);
		\draw(v1)--(z2);
		\draw(v1)--(z3);
		\draw(v1)--(z4);
		\draw(v2)--(z2);
		\draw(v2)--(z3);
		\draw(v2)--(z4);
		\draw(v3)--(z3);
		\draw(v3)--(z4);
		\draw(v4)--(z4);

		\draw(w1)--(z1);
		\draw(w1)--(z2);
		\draw(w1)--(z3);
		\draw(w1)--(z4);
		\draw(w2)--(z2);
		\draw(w2)--(z3);
		\draw(w2)--(z4);
		\draw(w3)--(z2);
		\draw(w3)--(z4);
		\draw(w4)--(z4);

    
\end{tikzpicture}
\caption{A twisted chain graph of order $2$. For instance,
  $3=2(2-1)+1=n(x-1)+y$ and $v_3$ is adjacent to $z_{2,1}$ and
  $z_{2,2}$, but not adjacent to $z_{1,1}$ and $z_{1,2}$.}
  \label{fig:twistedchain}
\end{figure}

\begin{definition}\label{def:twi}
  For $n\in \N$, a graph on the set of $3n^2$ vertices
  $A\cup B\cup C$, where $A=\{v_1, \ldots, v_{n^2}\}$,
  $B=\{w_1, \ldots, w_{n^2}\}$, and
  $C=\{z_{(i,j)}\colon 1\le i,j\le n\}$, is called a \emph{twisted
    chain graph} of order $n$ if
  \begin{itemize}
  \item for integers $x,y, i, j\in \{1, \ldots, n\}$ and $k=n(x-1)+y$,
    $v_k$ is adjacent to $z_{(i,j)}$ if and only if either ($x<i$) or
    ($x=i$ and $y\le j$);
  \item for integers $x,y, i, j\in \{1, \ldots, n\}$ and $k=n(x-1)+y$,
    $w_k$ is adjacent to $z_{(i,j)}$ if and only if either ($x<j$) or
    ($x=j$ and $y\le i$);
  \item the edge relation within $A\cup B$ and within $C$ is
    arbitrary.
  \end{itemize}
\end{definition}

\noindent We first show that a large twisted chain graph has large
rank-width.  We remark that a similar construction based on merging
two chain graphs in a mixed order can be found in Brandst\"adt et
al~\cite{Brandstadt2006}. Also, a slightly more general construction
was given by Dabrowski and Paulusma~\cite{DabrowskiP16}. Obtaining any
lower bound seems to follow from a careful examination and
modification of the constructions given in \cite{Brandstadt2006} or
\cite{DabrowskiP16}; however, we prefer to give our own direct proof
for the sake of completeness. Also, in those papers, the authors
provided a lower bound of clique-width, and its direct application to
rank-width does not provide a linear lower bound.

\begin{lemma}\label{lem:twistedgrid-rk}
  For every integer $n>0$, every twisted chain graph of order $12n$
  has rank-width at least~$n$.
\end{lemma}

Before we proceed to the proof of \cref{lem:twistedgrid-rk}, we need
to introduce some basic tools. A vertex bipartition $(X,Y)$ of a graph
$G$ is \emph{balanced with respect to a set $C\subseteq V(G)$} if
$\frac{|C|}{3}\le |X\cap C|, |Y\cap C|$.  We will need the following
standard fact.

\begin{lemma}\label{lem:balancedpartition}
  If $G$ is a graph of rank-width at most $w$ and $C\subseteq V(G)$
  with $|C|\geq 3$, then $G$ admits a vertex bipartition $(X,Y)$ with
  $\cutrk_G(X)\le w$ that is balanced with respect to $C$.
\end{lemma}
\begin{proof}
  Let $(T,L)$ be a rank-decomposition of $G$ of width at most $w$.  We
  subdivide an edge of $T$, and regard the new vertex as a root node.
  For each node $t\in V(T)$, let $\mu(t)$ be the number of leaves of
  $T$ that are descendants of $t$ and correspond to vertices of $C$.
  Now, we choose a node $t$ that is farthest from the root node and
  such that $\mu(t)\geq \frac{|C|}{3}$.  By the choice of $t$, either
  $t$ is a leaf or for each child $t'$ of $t$ we have
  $\mu(t')<\frac{|C|}{3}$.  Therefore, since $|C|\geq 3$, in any case
  $\frac{|C|}{3}\leq \mu(t)< \frac{2|C|}{3}$.  Let~$e$ be the edge
  connecting the node $t$ and its parent.  By the construction, the
  vertex bipartition associated with $e$ satisfies the required
  property.
\end{proof}

We now proceed to the proof of \cref{lem:twistedgrid-rk}.

\begin{proof}[of \cref{lem:twistedgrid-rk}]
  Let $m:=12n$ and let $G$ be a twisted chain graph of order $m$.
  Adopt the notation from Definition~\ref{def:twi} for $G$.  Suppose
  for the sake of contradiction that the rank-width of $G$ is smaller
  than $n$.  By \cref{lem:balancedpartition}, there exists a vertex
  bipartition $(S,T)$ of $G$ with $\cutrk_G(S)<n$ such that
  $|C\cap S|\geq \frac{|C|}{3}=m^2/3$ and similarly
  $|C\cap T|\geq m^2/3$.

  Suppose we have vertices $v_{a_1},\ldots,v_{a_k}\in A\cap S$ and
  $z_{(b_1,c_1)},\ldots,z_{(b_k,c_k)}\in C\cap T$ with the following
  property satisfied:
  $$a_1\leq(b_1-1)m+c_1< a_2\le (b_2-1)m+c_2< \cdots <a_k\le (b_k-1)m+c_k.$$
  Such a structure will be called an {\em{$A$-ordered
      $(S,T)$-matching}} of order $k$.  By the definition of adjacency
  in $G$ it follows that the submatrix of $A_G[S,T]$ induced by rows
  corresponding to vertices~$v_{a_i}$ and columns corresponding to
  vertices~$z_{(b_i,c_i)}$ has ones in the upper triangle and on the
  diagonal, and zeroes in the lower triangle. The rank of this
  submatrix is $k$, so since $\cutrk_G(A)<n$, there is no $A$-ordered
  $(S,T)$-matching of order $n$.  We similarly define the $A$-ordered
  $(T,S)$-matching of rank $n$, where the vertices $v_{a_i}$ belong to
  $T$ and vertices $z_{(b_i,c_i)}$ belong to $S$.  Likewise, there is
  no $A$-ordered $(T,S)$-matching of order $n$.

  Suppose now we have vertices $w_{a_1},\ldots,w_{a_k}\in B\cap S$ and
  $z_{(b_1,c_1)},\ldots,z_{(b_k,c_k)}\in C\cap T$ with the following
  property satisfied:
  $$a_1\leq(c_1-1)m+b_1< a_2\le (c_2-1)m+b_2< \cdots <a_k\le (c_k-1)m+b_k.$$
  Such a structure will be called a {\em{$B$-ordered
      $(S,T)$-matching}} of order $k$, and a {\em{$B$-ordered
      $(T,S)$-matching}} of order $k$ is defined analogously.  Again,
  the same reasoning as above shows that there is no $B$-ordered
  $(S,T)$-matching of order $n$, and no $B$-ordered $(T,S)$-matching
  of order $n$.

  Let $\preceq_1$ and $\preceq_2$ be lexicographic orders on
  $\{1,\ldots,m\}\times\{1,\ldots,m\}$, with the leading coordinate
  being the first one for $\preceq_1$ and the second for $\preceq_2$.

\pagebreak
\begin{claim}\label{cl:A-side}
  If there is a sequence
  $(x_1,y_1)\prec_1 (x_2,y_2)\prec_1\ldots \prec_1 (x_{4k},y_{4k})$
  such that $z_{(x_j,y_j)}\in S$ for odd $j$ and $z_{(x_j,y_j)}\in T$
  for even $j$, then there is either an $A$-ordered $(S,T)$-matching
  of order $k$, or an $A$-ordered $(T,S)$-matching of order~$k$.
\end{claim}
\begin{clproof}
Denote $r_j=(x_j-1)m+y_j$ for $j=1,2,\ldots,4k$.
For $i=1,2,\ldots,2k$, we define $a_i$ and $(b_i,c_i)$ as follows:
\begin{itemize}
\item if $v_{r_{2i-1}}\in S$, then $a_i=r_{2i-1}$ and
  $(b_i,c_i)=(x_{2i}, y_{2i})$, and
\item if $v_{r_{2i-1}}\in T$, then $a_i=r_{2i-1}$ and
  $(b_i,c_i)=(x_{2i-1}, y_{2i-1})$.
\end{itemize}
It follows that vertices $v_{a_1},\ldots,v_{a_{2k}}$ and
$z_{(b_1,c_1)},\ldots,z_{(b_{2k},c_{2k})}$ satisfy
$$a_1\leq(b_1-1)m+c_1< a_2\le (b_2-1)m+c_2< \cdots <a_k\le (b_k-1)m+c_k,$$
and for each $i$ we have that $v_{a_i}$ and $z_{(b_i,c_i)}$ belong to
different sides of the bipartition $(S,T)$.  Now, if for at least $k$
indices $i$ we have $v_{a_i}\in S$ and $z_{(b_i,c_i)}\in T$, then the
subsequence induced by elements with these indices gives an
$A$-ordered $(S,T)$-matching of order $k$.  Otherwise, there are at
least $k$ indices $i$ with $v_{a_i}\in T$ and $z_{(b_i,c_i)}\in S$,
and the subsequence induced by them gives an $A$-ordered
$(T,S)$-matching of order $k$.
\end{clproof}

A symmetric proof gives the following.

\begin{claim}\label{cl:B-side}
  If there is a sequence
  $(x_1,y_1)\prec_2 (x_2,y_2)\prec_2\ldots \prec_2 (x_{4k},y_{4k})$
  such that $z_{(x_j,y_j)}\in S$ for odd $j$ and $z_{(x_j,y_j)}\in T$
  for even $j$, then there is either a $B$-ordered $(S,T)$-matching of
  order $k$, or a $B$-ordered $(T,S)$-matching of order~$k$.
\end{claim}

From \cref{cl:A-side} and \cref{cl:B-side} it follows that the largest
possible length of sequences as in the statements is smaller than
$4n$.  For $i,j\in \{1,\ldots,m\}$, call the sets
$\{z_{(i,y)}\colon y\in \{1,\ldots,m\}\}$ and
$\{z_{(x,j)}\colon x\in \{1,\ldots,m\}\}$ the {\em{$i$th row}} and the
{\em{$j$th column}}, respectively.  A row or a column is {\em{mixed}}
if it contains both elements of~$S$ and elements of $T$.  Observe that
if there were at least $4n$ mixed rows, then by choosing vertices from
$S$ and $T$ alternately from these rows we would obtain a sequence of
length $4n$ as in the statement of \cref{cl:A-side}. Then
\cref{cl:A-side} gives us an $A$-ordered $(S,T)$-matching of order $n$
or an $A$-ordered $(T,S)$-matching of order~$n$, a contradiction.
Hence, there are less than $4n$ mixed rows, and symmetrically there
are less than $4n$ mixed columns.

Observe that if there were two non-mixed rows such that the first one
was contained in $S$ while the second was contained in $T$, then all
the $12n$ columns would be mixed, a contradiction.  Hence, either all
the non-mixed rows belong to~$S$, or all of them belong to
$T$. However, there are more than $12n-4n=8n=2m/3$ non-mixed rows, so
either $S$ or $T$ contains less than a third of vertices of $C$.  This
is a contradiction with the assumption that $(S,T)$ is balanced with
respect to~$C$.
\end{proof}

We now observe that if a graph class contains arbitrarily
large twisted chain graphs, then it does not admit low rank-width
colorings.

\begin{theorem}\label{thm:interval-permut}
  Let $\CCC$ be a hereditary graph class, and suppose for each
  positive integer $n$ some twisted chain graph of order $n$ belongs
  to $\CCC$.  Then $\CCC$ does not admit low rank-width colorings.
\end{theorem}
\begin{proof}
  We show that for every pair of integers $m\ge 3$ and $n\ge 1$, there
  is a graph $G\in \CCC$ such that for every coloring of $G$ with $m$
  colors, there is an induced subgraph $H$ that receives at most $3$
  colors and has rank-width at least~$n$.  This implies that $\CCC$
  does not admit low rank-width colorings.  We will need the following
  simple Ramsey-type argument.  \cref{cl:prod-rams} follows, e.g.,
  from~\cite[Theorem~11.5]{Trotter1992}, but we give a simple proof
  for the sake of completeness.
  \begin{claim}\label{cl:prod-rams}
    For all positive integers $k,d$, there exists an integer
    $M=M(k,d)$ such that for all sets $X,Y$ with $|X|=|Y|=M$ and all
    functions $f\colon X\times Y\to \{1,\ldots,d\}$, there exist
    subsets $X'\subseteq X$ and $Y'\subseteq Y$ with $|X'|=|Y'|=k$
    such that $f$ sends all elements of $X'\times Y'$ to the same
    value.
  \end{claim}
  \begin{clproof}
    We prove the claim for $M(k,d)=k\cdot d^{dk}$. Let $X_0$ be an
    arbitrary subset $X$ of size $dk$. For each $y\in Y$, define the
    {\em{type}} of $y$ as the function
    $g_y\colon X_0\to \{1,\ldots,d\}$ defined as
    $g_y(x)=f(x,y)$. There are at most $d^{dk}$ different types
    possible, so there is a subset $Y'\subseteq Y$ of size $k$ such
    that each element of $Y'$ has the same type $g$.  Since
    $|X_0|=dk$, there is some $i\in \{1,\ldots,d\}$ such that $g$
    yields value $i$ for at least $k$ elements of $X_0$. Then if we
    take $X'$ to be an arbitrary set of $k$ elements mapped to $i$ by
    $g$, then $f(x,y)=i$ for each $(x,y)\in X'\times Y'$, as required.
  \end{clproof}

  Let $M_1:=M(12n, m)$, $M_2:=M(M_1,m)$, and $M_3:=M(M_2,m)$.  Let
  $G\in \cal C$ be a twisted chain graph of order $M_3$; adopt the
  notation from Definition~\ref{def:twi} for~$G$.  Suppose $G$ is
  colored by $m$ colors by a coloring $c$.  By \cref{cl:prod-rams},
  there exist $X_1, Y_1\subseteq \{1, \ldots, M_3\}$ with
  $|X_1|=|Y_1|=M_2$ such that
  $\{z_{(x,y)}\colon (x,y)\in X_1\times Y_1\}$ is monochromatic under
  $c$.

  Now, for an index $k\in \{1,\ldots,M_3^2\}$, let
  $(i_1(k),j_1(k))\in \{1,\ldots,m\}\times \{1,\ldots,m\}$ be the
  unique pair such that $k=(i_1(k)-1)M_3+j_1(k)$, and let
  $(i_2(k),j_2(k))\in \{1,\ldots,m\}\times \{1,\ldots,m\}$ be the
  unique pair such that $k=(j_2(k)-1)M_3+i_2(k)$.  By reindexing
  vertices $A$ and $C$ using pairs $(i_1(k),j_1(k))$ and
  $(i_2(k),j_2(k))$, we may view coloring $c$ on $A$ and $C$ as a
  coloring on $\{1,\ldots,M_3\}\times \{1,\ldots,M_3\}$.  By applying
  \cref{cl:prod-rams} to the vertices from $A$ indexed by
  $X_1\times Y_1$, we obtain subsets $X_2\subseteq X_1$ and
  $Y_2\subseteq Y_1$ such that $|X_2|=|Y_2|=M_1$ and the set
  $\{v_{(x-1)M_3+y}\colon x\in X_2, y\in Y_2\}$ is monochromatic.
  Finally, by applying \cref{cl:prod-rams} to the vertices from~$B$
  indexed by $X_2\times Y_2$, we obtain subsets $X_3\subseteq X_2$ and
  $Y_3\subseteq Y_2$ such that $|X_3|=|Y_3|=12n$ and the set
  $\{w_{(y-1)M_3+x}\colon (x,y)\in X_3\times Y_3\}$ is monochromatic.
  Now observe that the subgraph
  $G[\{v_{(x-1)M_3+y}, w_{(y-1)M_3+x}, z_{(x,y)}\colon (x,y)\in
  X_3\times Y_3\}]$
  receives at most $3$ colors, and is a twisted chain graph of order
  $12n$.  By \cref{lem:twistedgrid-rk} it has rank-width at least $n$,
  so this proves the claim.
\end{proof}

We now observe that a twisted chain graph of order $n$ is an interval
graph, provided each of $A$, $B$, and $C$ is a clique, and there are
no edges between $A$ and $B$.  Similarly, for each $n$ there is a
twisted chain graph of order $n$ that is a permutation graph.  See
Figures~\ref{fig:interval} and~\ref{fig:permutation} for examples of
intersection models.

\begin{figure}
\centerline{\includegraphics[scale=0.45]{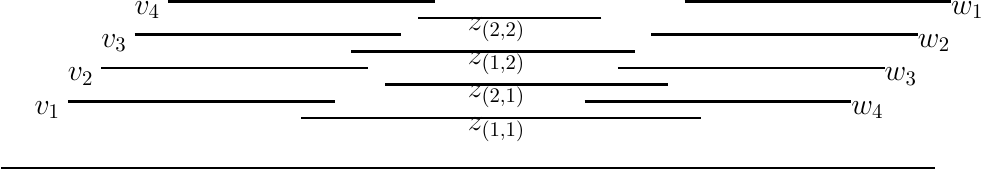}}
\caption{An interval intersection model of a twisted chain graph of
  order $2$. }
\label{fig:interval}
\end{figure}

\begin{figure}
\centerline{\includegraphics[scale=0.45]{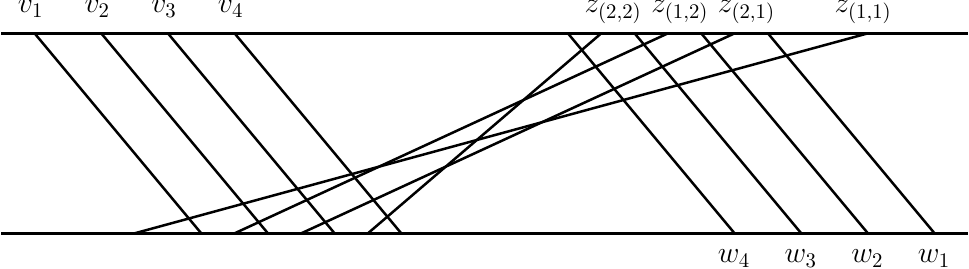}}
\caption{A permutation intersection model of a twisted chain graph of
  order $2$. }
\label{fig:permutation}
\end{figure}

\begin{lemma}\label{lem:twistedgrid1}
  Let $G$ be a twisted chain graph of order~$n$, for some positive
  integer~$n$.  If each of $A$, $B$, and $C$ is a clique, and there
  are no edges between $A$ and $B$, then $G$ is an interval graph.
\end{lemma}
\begin{proof}
Let $M>2n^2$. Consider the following interval model:
\begin{itemize}
\item For each $i\in \{1,\ldots,n^2\}$, assign interval $[0,i]$ to
  $v_i$.
\item For each $i\in \{1,\ldots,n^2\}$, assign interval $[M-i,M]$ to
  $w_i$.
\item For each $x,y\in \{1,\ldots,n\}\times \{1,\ldots,n\}$, assign
  interval $[(x-1)n+y,M-(y-1)n-x]$ to $z_{(x,y)}$.
\end{itemize}
It can be easily seen that this is an interval model of the twisted
chain graph of order $n$.
\end{proof}

\begin{lemma}\label{lem:twistedgrid2}
  For each positive integer~$n$, there is a twisted chain graph of
  order~$n$ that is a permutation graph.
\end{lemma}
\begin{proof}
  Let $M>10n^2$. Consider the following permutation model, spanned
  between horizontal lines $\ell_1$ with $y$-coordinate $0$ and
  $\ell_2$ with $y$-coordinate $1$:
\begin{itemize}
\item For each $i\in \{1,\ldots,n^2\}$, assign the segment with
  endpoints $(i,0)$ and $(i,1)$ to $v_i$.
\item For each $i\in \{1,\ldots,n^2\}$, assign the segment with
  endpoints $(M-i,0)$ and $(M-i,1)$ to $w_i$.
\item For each $x,y\in \{1,\ldots,n\}\times \{1,\ldots,n\}$, assign
  the segment with endpoints $((x-1)n+y,0)$ and $(M-(y-1)n-x,1)$ to
  $z_{(x,y)}$
\end{itemize}
It can be easily seen that this is a permutation model of some twisted
chain graph of order $n$.
\end{proof}

By \cref{thm:interval-permut}, we obtain the following.

\begin{theorem}
  The classes of interval graphs and permutation graphs do not admit
  low rank-width colorings.
\end{theorem}

\section{Erd\H{o}s-Hajnal property and $\chi$-boundedness}\label{sec:EHchi}
\vspace{-2mm}

\noindent A hereditary graph class $\CCC$ has the \emph{Erd\H{o}s-Hajnal
  property} if there is \mbox{$\epsilon>0$}, depending only on $\CCC$, such
that every $n$-vertex graph in $\CCC$ has either an independent set or
a clique of size $n^{\epsilon}$.  The conjecture of Erd\H{o}s and
Hajnal~\cite{ErdosH1989} states that for every fixed graph $H$, the
class of graphs not having $H$ as an induced subgraph has the
Erd\H{o}s-Hajnal property; cf.~\cite{Chudnovsky2014}.  We prove that
every class admitting low rank-width colorings has the
Erd\H{o}s-Hajnal property.  The proof is based on the fact that every
class of graphs of bounded rank-width has this property, which was
shown by Oum and Seymour~\cite{OumSpersonal}.  Since this claim is not
written in any published work, we include the proof for the
completeness.
	
A graph is a \emph{cograph} if it can be recursively constructed from
isolated vertices by means of the following two operations: (1) taking
disjoint union of two graphs and (2) joining two graphs, i.e., taking
their disjoint union and adding all possible edges with one endpoint
in the first graph and the second endpoint in the second.  It is
well-known that every $n$-vertex cograph contains either an
independent set or a clique of size $n^{1/2}$; this follows from the
fact that cographs are perfect.

\begin{lemma}[Oum and Seymour~\cite{OumSpersonal}]\label{lem:ehbddrw}
  For every positive integer $p$, there exists a constant
  $\delta=\delta(p)$ such that every $n$-vertex graph of rank-width at
  most~$p$ contains either an independent set or a clique of size at
  least $n^\delta$.
\end{lemma}
\vspace{-2mm}
\begin{proof}
  Let $\kappa(p):=\frac{1}{\log_2 3+ p}$ and
  $\delta(p):=\frac{\kappa(p)}{2}$.  We first prove by induction on
  $n$ that every $n$-vertex graph of rank-width at most $p$ contains a
  cograph of size at least $n^{\kappa(p)}$.  Assume $n\geq 3$, for
  otherwise every graph on at most $2$ vertices is a cograph.  Let $G$
  be an $n$-vertex graph of rank-width at most $p$. By
  \cref{lem:balancedpartition}, $G$ has a vertex bipartition $(A,B)$
  where $\cutrk_G(A)\le p$ and $\abs{A}, \abs{B}\geq \frac{n}{3}$.
  Since $\cutrk_G(A)\le p$, there exist $A'\subseteq A$ and
  $B'\subseteq B$ with $\abs{A'}, \abs{B'}\geq \frac{n}{3\cdot 2^p}$,
  such that either there are no edges between $A'$ and $B'$, or every
  vertex in $A'$ is adjacent to every vertex in $B'$.  By induction
  hypothesis, $G[A']$ contains a cograph $H_1$ and $G[B']$ contains a
  cograph $H_2$, both as induced subgraphs, such that
  $\abs{V(H_1)}, \abs{V(H_2)}\ge (\frac{n}{3\cdot 2^p})^{\kappa(p)}$.
  Note that $G[V(H_1)\cup V(H_2)]$ is a cograph.  Since
  $\kappa(p)=\frac{1}{\log_2 3+ p}$, we have
  $\abs{V(H_1)}+\abs{V(H_2)}\ge 2(\frac{n}{3\cdot 2^p})^{\kappa(p)}=
  n^{\kappa(p)}$.
  We conclude that~$G$ contains either an independent set or a clique
  of size at least $(n^{\kappa(p)})^{1/2}=n^{\delta(p)}$.
\end{proof}

\vspace{-2mm}
\begin{proposition}\label{prop:lowrwEH}
  Every hereditary class of graphs that admits low rank-width
  colorings has the Erd\H{o}s-Hajnal property.
\end{proposition}
\vspace{-2mm}
\begin{proof}
  Let $\CCC$ be the class of graphs in question.  Fix any $G\in \CCC$,
  say on~$n$ vertices.  Since $\CCC$ admits low rank-width colorings,
  there exist functions \mbox{$N\colon \N\rightarrow\N$} and
  $R\colon \N\rightarrow \N$, depending only on $\CCC$, such that for
  all $p$, $G$ can be colored using $N(p)$ colors so that each induced
  subgraph of $G$ that receives $i\le p$ colors has rank-width at most
  $R(i)$.  Let $c$ be such a coloring for $p=1$; then $c$ uses $N(1)$
  colors.
  Let $\delta(p)$ be the function defined in \cref{lem:ehbddrw}.  We
  define
  \[\epsilon:=\min \left( \frac{\delta(R(1))}{2}, \frac{1}{2\log_2
      N(1)}\right),\]
  and claim that $G$ has an independent set or a clique of size at
  least $n^{\epsilon}$.

  First, assume that $n\ge N(1)^2$.  Then there is a color $i$ such that
  $|c^{-1}(i)|\geq \frac{n}{N(1)}$.  Thus, the subgraph~$H$ induced by
  the vertices with color $i$ has at least $\frac{n}{N(1)}$ vertices
  and rank-width at most $R(1)$.  Since $n\ge N(1)^2$, by
  \cref{lem:ehbddrw}, $H$, and thus also $G$, contains either an
  independent set or a clique of size at least
  $$\abs{V(H)}^{\delta(R(1))}\ge \left(\frac{n}{N(1)}
  \right)^{\delta(R(1))}\ge n^{\frac{\delta(R(1))}{2}}\ge
  n^{\epsilon}.$$
	
  Second, assume $n<N(1)^2$.  Then
  $n^{\epsilon}<N(1)^{2\epsilon}\leq 2$, so any two-vertex induced
  subgraph of $G$ is either a clique or an independent set of size at
  least~$n^{\epsilon}$.
\end{proof}
	
A hereditary class $\CCC$ of graphs is \emph{$\chi$-bounded} if there
exists a function \mbox{$f\colon \mathbb N\rightarrow \mathbb N$} such that
for every $G\in \CCC$ and an induced subgraph $H$ of $G$, we have
$\chi(H)\le f(\omega(H))$, where $\chi(H)$ is the chromatic number of
$H$ and $\omega(H)$ is the size of a maximum clique in $H$.  It was
proved by Dvo{\v{r}}{\'a}k and Kr{\'a}l'~\cite{DvorakK2012} that for
every $p$, the class of graphs of rank-width at most $p$ is
$\chi$-bounded.
		
\begin{theorem}[Dvo\v{r}\'{a}k and
  Kr\'{a}l'~\cite{DvorakK2012}]\label{thm:bddrwchi}
  For each positive integer $p$, the class of graphs of rank-width at
  most $p$ is $\chi$-bounded.
\end{theorem}

We observe that this fact directly generalizes to classes admitting
low rank-width colorings.

\begin{proposition}\label{prop:lowrwXbdd}
  Every hereditary class of graphs that admits low rank-width
  colorings is $\chi$-bounded.
\end{proposition}
\begin{proof}
  Let $\CCC$ be the hereditary class of graphs in question, let
  $G\in\CCC$, and let $q=\omega(G)$.  Since $\CCC$ admits low
  rank-width colorings, there exist functions
  $N\colon \N\rightarrow\N$ and $R\colon \N\rightarrow \N$, depending
  only on $\CCC$, such that for all $p$, $G$ can be colored using
  $N(p)$ colors so that each induced subgraph of $G$ that receives
  $i\le p$ colors has rank-width at most $R(i)$.  Let $c$ be such a
  coloring for $p=1$, and w.l.o.g. suppose that~$c$ uses colors
  $\{1,\ldots,N(1)\}$.  By \cref{thm:bddrwchi}, there is function
  $f_{R(1)}$ such that for every graph $H$ of rank-width at most
  $R(1)$, we have $\chi(H)\le f_{R(1)}(\omega(H))$.
	
  For $i\in \{1,\ldots,N(1)\}$, let $G_i=G[c^{-1}(i)]$ be the subgraph
  induced by vertices of color $i$.  Since~$G_i$ is an induced
  subgraph of $G$, we have that $\omega(G_i)\leq q$, so $G_i$ has a
  proper coloring $c_i$ using $f_{R(1)}(q)$ colors, say colors
  $\{1,\ldots,f(q)\}$.  Then we can take the product coloring $c'$ of
  $G$ defined as $c'(u)=(c(u),c_{c(u)}(u))$. Observe that since each
  coloring $c_i$ is proper, $c'$ is a proper coloring of $G$.  It
  follows that the chromatic number of $G$ is at most
  $N(1)\cdot f(q)$, so we can take $f'(q)=N(1)\cdot f_{R(1)}(q)$ as
  the $\chi$-bounding function for $\CCC$.
\end{proof}

\vspace{-7mm}
\section{Conclusions}\label{sec:conc}
\vspace{-2mm}

\noindent In this work we introduced the concept of low rank-width colorings,
and showed that such colorings exist on $r$th powers of graphs from
any bounded expansion class, for any fixed $r$, as well as on unit
interval and bipartite permutation graphs.  These classes are
non-sparse and have unbounded rank-width.  On the negative side, the
classes of interval and permutation graphs do not admit low rank-width
colorings.

The obvious open problem is to characterise hereditary graph classes
which admit low rank-width colorings in the spirit of the
characterisation theorem for graph classes admitting low tree-depth
colorings. We believe that \cref{thm:interval-permut} may provide some
insight into this question, as it shows that containing arbitrarily
large twisted chain graphs is an obstacle for admitting low rank-width
colorings. Is it true that every hereditary graph class that does not
admit low rank-width colorings has to contain arbitrarily large
twisted chain graphs?

In this work we did not investigate the question of computing low
rank-width colorings, and this question is of course crucial for any
algorithmic applications.  Our proof for the powers of sparse graphs
can be turned into a polynomial-time algorithm that, given a graph $G$
from a graph class of bounded expansion $\CCC$, first computes a low
tree-depth coloring, and then turns it into a low rank-width coloring
of $G^r$, for a fixed constant $r$.  However, we do not know how to
efficiently compute a low rank-width coloring given the graph $G^r$
alone, without the knowledge of $G$.  The even more general problem of
efficiently constructing an approximate low rank-width coloring of any
given graph remains wide open.

Finally, we remark that our proof for the existence of low rank-width
colorings on powers of graphs from a class of bounded expansion
actually yields a slightly stronger result. Ganian et
al.~\cite{GanianHNJP2012} introduced a parameter {\em{shrub-depth}}
(or {\em{SC-depth}}), which is a depth analogue of rank-width, in the
same way as tree-depth is a depth analogue of tree-width. DeVos, Kwon
and Oum~\cite{DevosKO2019} later introduced another parameter
{\em{rank-depth}} which is tied to the shrub-depth.  It can be shown
that for constant~$r$, the $r$th power of a graph of constant
tree-depth belongs to a class of constant shrub-depth, and hence our
colorings for powers of graphs from a class of bounded expansion are
actually low shrub-depth colorings. In a follow-up of this
work~\cite{Gaj18} we proved that in fact every first-order
interpretation of a bounded expansion class (the $r$th power of a
graph is a simple first-order interpretation) has low shrub-depth
colorings.  It remains an important open problem whether the $r$th
powers, or more generally first-order interpretations, of nowhere
dense graph classes for any $\epsilon>0$ admit low shrub-depth
colorings with $\Oof(n^\epsilon)$ colors.

\vskip 0.2cm \textbf{Acknowledgment.} We would like to thank
Konrad Dabrowski for pointing out the known constructions similar to
twisted chain graphs. We also thank the anonymous
reviewers for careful reading and helpful advices to improve the
original manuscript.

\bibliographystyle{abbrv}
\bibliography{ref}

\begin{thebibliography}{10}

\bibitem{bodlaender1998rankings}
H.~L. Bodlaender, J.~S. Deogun, K.~Jansen, T.~Kloks, D.~Kratsch, H.~M{\"u}ller,
  and Z.~Tuza.
\newblock Rankings of graphs.
\newblock {\em SIAM Journal on Discrete Mathematics}, 11(1):168--181, 1998.

\bibitem{BonamyBT2016}
M.~Bonamy, N.~Bousquet, and S.~Thomass\'{e}.
\newblock {The {E}rd\H{o}s-{H}ajnal conjecture for long holes and antiholes}.
\newblock {\em SIAM J. Discrete Math.}, 30(2):1159--1164, 2016.

\bibitem{Brandstadt2006}
A.~Brandst\"{a}dt, J.~Engelfriet, H.-O. Le, and V.~V. Lozin.
\newblock Clique-width for 4-vertex forbidden subgraphs.
\newblock {\em Theory Comput. Syst.}, 39(4):561--590, 2006.

\bibitem{bui2011boolean}
B.-M. Bui-Xuan, J.~A. Telle, and M.~Vatshelle.
\newblock Boolean-width of graphs.
\newblock {\em Theoretical Computer Science}, 412(39):5187--5204, 2011.

\bibitem{ChenGP02}
Z.~Chen, M.~Grigni, and C.~H. Papadimitriou.
\newblock Map graphs.
\newblock {\em J. {ACM}}, 49(2):127--138, 2002.

\bibitem{Chudnovsky2014}
M.~Chudnovsky.
\newblock {The {E}rd\H{o}s-{H}ajnal conjecture---a survey}.
\newblock {\em Journal of Graph Theory}, 75(2):178--190, 2014.

\bibitem{Chudnovsky2018}
M.~Chudnovsky and S.-i. Oum.
\newblock {Vertex-minors and the Erd\H{o}s--Hajnal conjecture}.
\newblock {\em Discrete Math.}, 341(12):3498--3499, 2018.

\bibitem{courcelle1990graph}
B.~Courcelle.
\newblock Graph rewriting: {A}n algebraic and logic approach.
\newblock {\em Handbook of Theoretical Computer Science, volume B}, pages
  193--242, 1990.

\bibitem{courcelle1990monadic}
B.~Courcelle.
\newblock The monadic second-order logic of graphs. {I}. {R}ecognizable sets of
  finite graphs.
\newblock {\em Information and Computation}, 85(1):12--75, 1990.

\bibitem{courcelle1993handle}
B.~Courcelle, J.~Engelfriet, and G.~Rozenberg.
\newblock Handle-rewriting hypergraph grammars.
\newblock {\em Journal of Computer and System Sciences}, 46(2):218--270, 1993.

\bibitem{CourcelleGK2011}
B.~Courcelle, C.~Gavoille, and M.~M. Kant\'{e}.
\newblock Compact labelings for efficient first-order model-checking.
\newblock {\em J. Comb. Optim.}, 21(1):19--46, 2011.

\bibitem{courcelle2000linear}
B.~Courcelle, J.~A. Makowsky, and U.~Rotics.
\newblock Linear time solvable optimization problems on graphs of bounded
  clique-width.
\newblock {\em Theory of Computing Systems}, 33(2):125--150, 2000.

\bibitem{DabrowskiP16}
K.~Dabrowski and D.~Paulusma.
\newblock Clique-width of graph classes defined by two forbidden induced
  subgraphs.
\newblock {\em Comput. J.}, 59(5):650--666, 2016.

\bibitem{deogun1994vertex}
J.~S. Deogun, T.~Kloks, D.~Kratsch, and H.~M{\"u}ller.
\newblock On vertex ranking for permutation and other graphs.
\newblock In {\em STACS 1994}, pages 747--758. Springer, 1994.

\bibitem{devos2004excluding}
M.~DeVos, G.~Ding, B.~Oporowski, D.~P. Sanders, B.~Reed, P.~Seymour, and
  D.~Vertigan.
\newblock Excluding any graph as a minor allows a low tree-width 2-coloring.
\newblock {\em Journal of Combinatorial Theory, Series B}, 91(1):25--41, 2004.

\bibitem{DevosKO2019}
M.~DeVos, O.~Kwon, and S.~Oum.
\newblock Branch-depth: Generalizing tree-depth of graphs.
\newblock {\em arXiv preprint}, 2019.
\newblock arXiv:1903.11988.

\bibitem{diestel2012graph}
R.~Diestel.
\newblock {\em Graph Theory, 4th Edition}, volume 173 of {\em Graduate {T}exts
  in {M}athematics}.
\newblock Springer, 2012.

\bibitem{dvovrak2013testing}
Z.~Dvo{\v{r}}{\'a}k, D.~Kr{\'a}l', and R.~Thomas.
\newblock Testing first-order properties for subclasses of sparse graphs.
\newblock {\em Journal of the ACM (JACM)}, 60(5):36, 2013.

\bibitem{DvorakK2012}
Z.~Dvo\v{r}\'{a}k and D.~Kr\'{a}l'.
\newblock Classes of graphs with small rank decompositions are $\chi$-bounded.
\newblock {\em European Journal of Combinatorics}, 33(4):679--683, 2012.

\bibitem{ErdosH1989}
P.~Erd\H{o}s and A.~Hajnal.
\newblock Ramsey-type theorems.
\newblock {\em Discrete Applied Mathematics}, 25(1-2):37--52, 1989.

\bibitem{FoxPS2017}
J.~Fox, J.~Pach, and A.~Suk.
\newblock {Erd\H{o}s-{H}ajnal conjecture for graphs with bounded
  {VC}-dimension}.
\newblock In {\em 33rd {I}nternational {S}ymposium on {C}omputational
  {G}eometry}, volume~77 of {\em LIPIcs. Leibniz Int. Proc. Inform.}, pages
  Art. No. 43, 15. Schloss Dagstuhl. Leibniz-Zent. Inform., Wadern, 2017.

\bibitem{Gaj18}
J.~Gajarsk{\'y}, S.~Kreutzer, J.~Nesetril, P.~O. de~Mendez, M.~Pilipczuk,
  S.~Siebertz, and S.~Torunczyk.
\newblock {First-Order Interpretations of Bounded Expansion Classes}.
\newblock In I.~Chatzigiannakis, C.~Kaklamanis, D.~Marx, and D.~Sannella,
  editors, {\em 45th International Colloquium on Automata, Languages, and
  Programming (ICALP 2018)}, volume 107 of {\em Leibniz International
  Proceedings in Informatics (LIPIcs)}, pages 126:1--126:14, Dagstuhl, Germany,
  2018. Schloss Dagstuhl--Leibniz-Zentrum fuer Informatik.

\bibitem{GanianHNJP2012}
R.~Ganian, P.~Hlin{\v{e}}n{\`y}, J.~Ne{\v{s}}et{\v{r}}il,
  J.~Obdr{\v{z}}{\'a}lek, P.~Ossona~de Mendez, and R.~Ramadurai.
\newblock When trees grow low: shrubs and fast {${\rm MSO}_1$}.
\newblock In {\em MFCS 2012}, volume 7464 of {\em Lecture Notes in Computer
  Science}, pages 419--430. Springer, 2012.

\bibitem{grohe2011methods}
M.~Grohe and S.~Kreutzer.
\newblock Methods for algorithmic meta theorems.
\newblock {\em Model Theoretic Methods in Finite Combinatorics}, 558:181--206,
  2011.

\bibitem{grohe2014deciding}
M.~Grohe, S.~Kreutzer, and S.~Siebertz.
\newblock Deciding first-order properties of nowhere dense graphs.
\newblock In {\em STOC 2014}, pages 89--98. ACM, 2014.

\bibitem{grohe2004learnability}
M.~Grohe and G.~Tur{\'a}n.
\newblock Learnability and definability in trees and similar structures.
\newblock {\em Theory of Computing Systems}, 37(1):193--220, 2004.

\bibitem{GurskiW2009}
F.~Gurski and E.~Wanke.
\newblock The {NLC}-width and clique-width for powers of graphs of bounded
  tree-width.
\newblock {\em Discrete Applied Mathematics}, 157(4):583--595, 2009.

\bibitem{Gyarfas1987}
A.~Gy\'{a}rf\'{a}s.
\newblock Problems from the world surrounding perfect graphs.
\newblock {\em Zastosowania Matematyki (Applicationes Mathematicae)},
  19:413--441, 1987.

\bibitem{hlinveny2008width}
P.~Hlin{\v{e}}n{\`y}, S.~Oum, D.~Seese, and G.~Gottlob.
\newblock Width parameters beyond tree-width and their applications.
\newblock {\em The Computer Journal}, 51(3):326--362, 2008.

\bibitem{kierstead2003orders}
H.~A. Kierstead and D.~Yang.
\newblock Orderings on graphs and game coloring number.
\newblock {\em Order}, 20(3):255--264, 2003.

\bibitem{kwon17}
O.~Kwon, M.~Pilipczuk, and S.~Siebertz.
\newblock On low rank-width colorings.
\newblock In {\em Graph-theoretic concepts in computer science}, volume 10520
  of {\em Lecture Notes in Comput. Sci.}, pages 372--385. Springer, Cham, 2017.

\bibitem{Liebenau2019}
A.~Liebenau, M.~Pilipczuk, P.~Seymour, and S.~Spirkl.
\newblock {Caterpillars in Erd\H{o}s--{H}ajnal}.
\newblock {\em J. Combin. Theory Ser. B}, 136:33--43, 2019.

\bibitem{Lozin2011}
V.~V. Lozin.
\newblock Minimal classes of graphs of unbounded clique-width.
\newblock {\em Annals of Combinatorics}, 15(4):707--722, 2011.

\bibitem{LozinR2007}
V.~V. Lozin and G.~Rudolf.
\newblock Minimal universal bipartite graphs.
\newblock {\em Ars Combinatoria}, 84:345--356, 2007.

\bibitem{nevsetvril2006tree}
J.~Ne{\v{s}}et{\v{r}}il and P.~Ossona~de Mendez.
\newblock Tree-depth, subgraph coloring and homomorphism bounds.
\newblock {\em European Journal of Combinatorics}, 27(6):1022--1041, 2006.

\bibitem{nevsetvril2008grad}
J.~Ne{\v{s}}et{\v{r}}il and P.~Ossona~de Mendez.
\newblock Grad and classes with bounded expansion {I}. {D}ecompositions.
\newblock {\em European Journal of Combinatorics}, 29(3):760--776, 2008.

\bibitem{nevsetvril2011nowhere}
J.~Ne{\v{s}}et{\v{r}}il and P.~Ossona~de Mendez.
\newblock On nowhere dense graphs.
\newblock {\em European Journal of Combinatorics}, 32(4):600--617, 2011.

\bibitem{nevsetvril2003order}
J.~Ne{\v{s}}et{\v{r}}il and S.~Shelah.
\newblock On the order of countable graphs.
\newblock {\em European Journal of Combinatorics}, 24(6):649--663, 2003.

\bibitem{oum2008rank}
S.~Oum.
\newblock Rank-width is less than or equal to branch-width.
\newblock {\em Journal of Graph Theory}, 57(3):239--244, 2008.

\bibitem{oum2016rank}
S.~Oum.
\newblock Rank-width: Algorithmic and structural results.
\newblock {\em Discrete Applied Mathematics}, 2016.

\bibitem{OumSpersonal}
S.~Oum and P.~Seymour.
\newblock Personal communication.

\bibitem{oum2006approximating}
S.~Oum and P.~Seymour.
\newblock Approximating clique-width and branch-width.
\newblock {\em Journal of Combinatorial Theory, Series B}, 96(4):514--528,
  2006.

\bibitem{robertson2003graph}
N.~Robertson and P.~D. Seymour.
\newblock Graph minors. {XVI}. {E}xcluding a non-planar graph.
\newblock {\em Journal of Combinatorial Theory, Series B}, 89(1):43--76, 2003.

\bibitem{schaffer1989optimal}
A.~A. Sch{\"a}ffer.
\newblock Optimal node ranking of trees in linear time.
\newblock {\em Information Processing Letters}, 33(2):91--96, 1989.

\bibitem{surveychi}
A.~Scott and P.~Seymour.
\newblock A survey of $\chi$-boundedness.
\newblock
  \url{https://web.math.princeton.edu/~pds/papers/chibounded/paper.pdf}, 2018.

\bibitem{Trotter1992}
W.~T. Trotter.
\newblock {\em Combinatorics and partially ordered sets}.
\newblock Johns Hopkins Series in the Mathematical Sciences. Johns Hopkins
  University Press, 1992.

\bibitem{wanke1994k}
E.~Wanke.
\newblock $k$-{NLC} graphs and polynomial algorithms.
\newblock {\em Discrete Applied Mathematics}, 54(2-3):251--266, 1994.

\bibitem{zhu2009coloring}
X.~Zhu.
\newblock Colouring graphs with bounded generalized colouring number.
\newblock {\em Discrete Mathematics}, 309(18):5562--5568, 2009.

\end{thebibliography}

\end{document}